\documentclass[11pt]{article}

\usepackage{graphicx}

\usepackage{amssymb}%
\usepackage{graphicx}
\newtheorem{theorem}{Theorem}
\newtheorem{definition}{Definition}
\newtheorem{lemma}{Lemma}
\newenvironment{proof}{\vspace{-2mm}
{\bf Proof:} \rm}{\mbox{} \hfill $\Box$ \vspace{1ex} }

\title{A Simple Optimal Binary Representation of Mosaic Floorplans and Baxter
Permutations}
\author{Bryan Dawei He\thanks{Department of Computer Science, 
California Institute of Technology, Pasadena, CA 91125. Email: bryanhe@caltech.edu}}
\date{}

\begin{document}
\maketitle

\begin{abstract}
A \emph{floorplan} is a rectangle subdivided into smaller rectangular
sections by horizontal and vertical line segments. Each section
in the floorplan is called a \emph{block}. Two floorplans are considered
equivalent if and only if there is a one-to-one correspondence
between the blocks in the two floorplans such that the
relative position relationship of the blocks in one floorplan is the same as
the relative position relationship of the corresponding
blocks in another floorplan. The objects of \emph{Mosaic floorplans}
are the same as floorplans, but an alternative definition of
equivalence is used. Two mosaic floorplans are considered equivalent
if and only if they can be converted to each other by sliding the
line segments that divide the blocks. 

Mosaic floorplans are widely used in VLSI circuit design.
An important problem in this area is to find short binary
string representations of the set of $n$-block mosaic floorplans. 
The best known representation is the \emph{Quarter-State Sequence}
which uses $4n$ bits.
This paper introduces a simple binary representation
of $n$-block mosaic floorplan using $3n-3$ bits.
It has been shown that any binary representation of $n$-block mosaic
floorplans must use at least $(3n-o(n))$ bits. 
Therefore, the representation presented in this paper is
optimal (up to an additive lower order term). 

\emph{Baxter permutations} are a set of permutations defined by prohibited
subsequences. Baxter permutations have been shown to have one-to-one
correspondences to many interesting objects in the so-called 
\emph{Baxter combinatorial family}. In particular,
there exists a simple one-to-one correspondence between
mosaic floorplans and Baxter permutations. As a result, the methods introduced
in this paper also lead to an optimal binary representation
of Baxter permutations and all objects in the Baxter combinatorial family.
\end{abstract}

\maketitle

\section{Introduction} \label{section: Introduction}

In this section, we introduce the definition of mosaic floorplans and 
Baxter permutations, describe their applications and previous work in
this area, and state our main result.

\subsection{Floorplans and Mosaic Floorplans} \label{subsection: Floorplans and Mosaic Floorplans}
\begin{definition} 
{
A \emph{floorplan} is a rectangle subdivided into smaller rectangular
subsections by horizontal and vertical line segments such that no
four subsections meet at the same point.
}
\end{definition}

\begin{figure}[htbp]
\centerline
{
\includegraphics[scale=0.1]{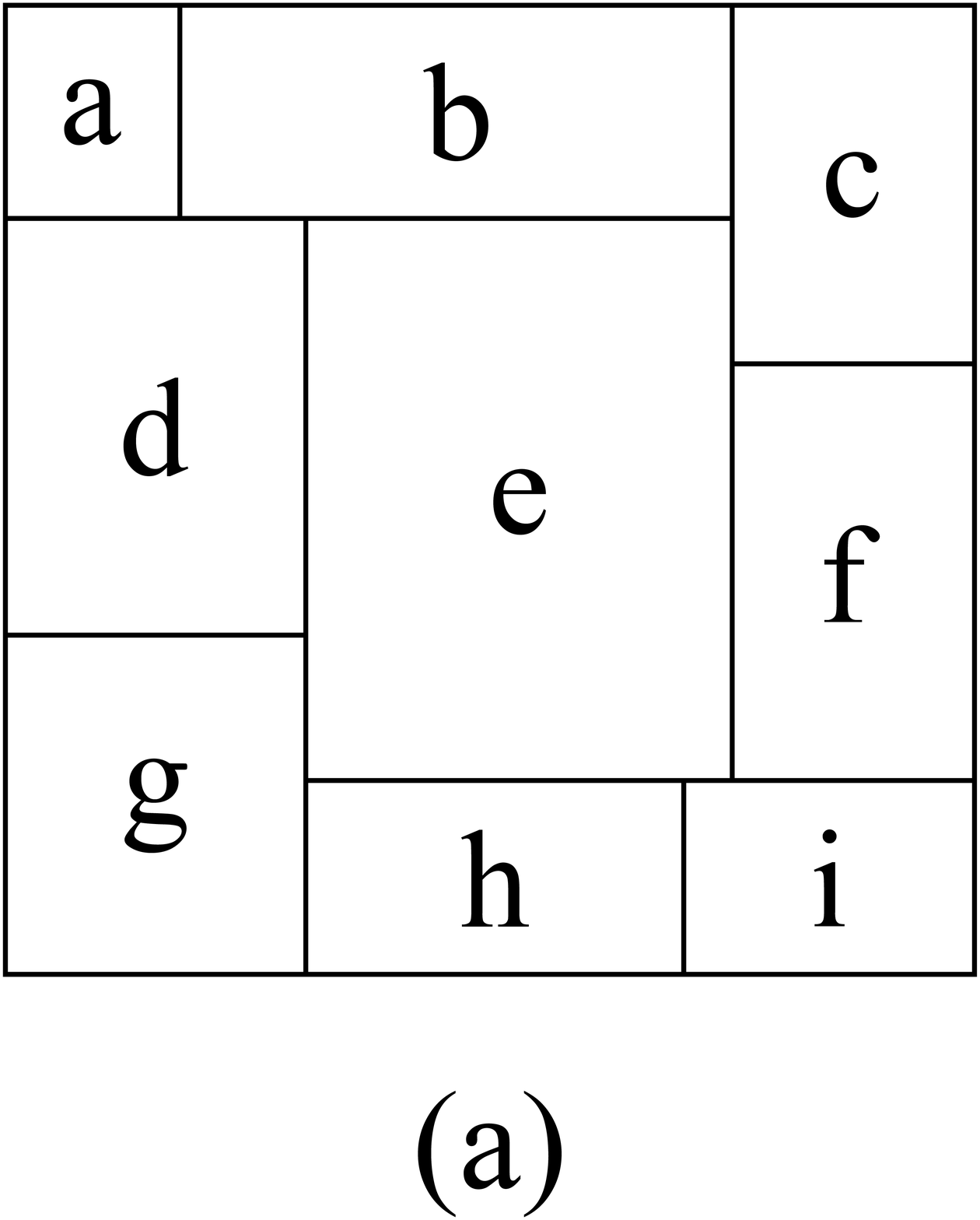}
\hspace{0.15in}
\includegraphics[scale=0.1]{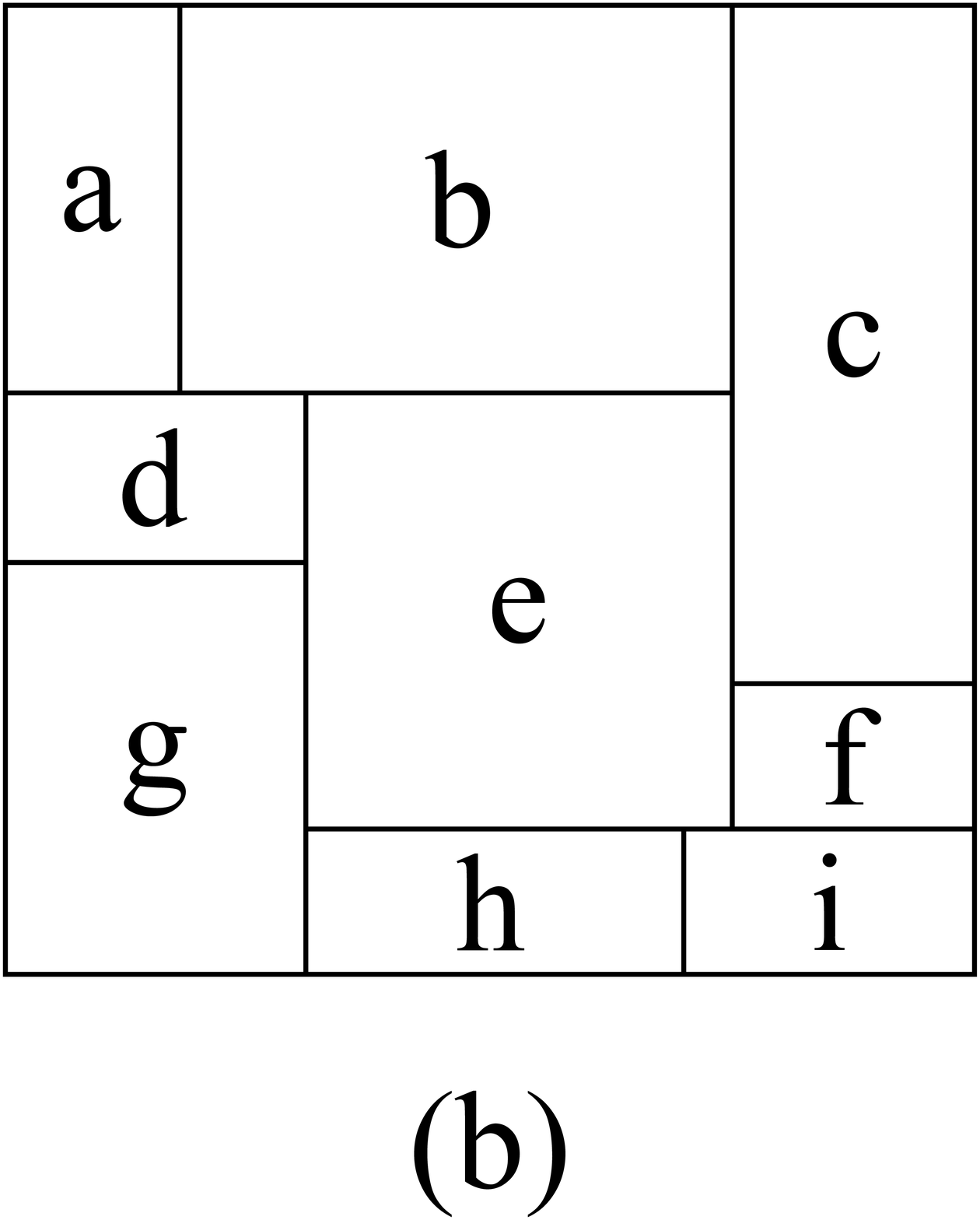}
\hspace{0.15in}
\includegraphics[scale=0.1]{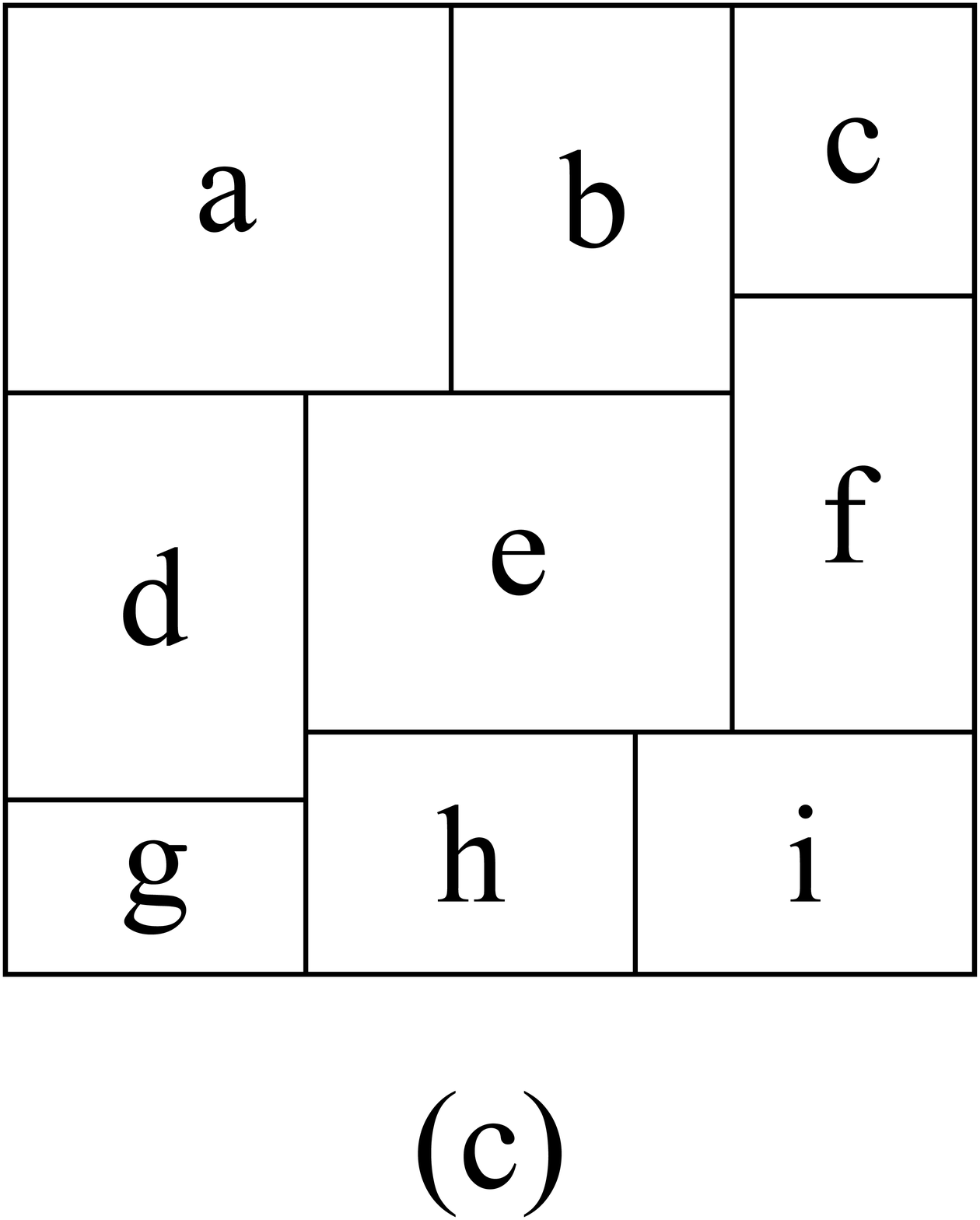}
}
\caption{Three example floorplans.}
\label{figure: example floorplans}
\end{figure}

The smaller rectangular subsections are called \emph{blocks}. 
Figure \ref{figure: example floorplans} shows three floorplans, each containing
9 blocks. Note that the horizontal and vertical line segments do not
cross each other. They can only form \emph{T-junctions}
($\vdash$, $\bot$, $\dashv$, and $\top$).


The definition of equivalent floorplans does not consider the size of the 
blocks of the floorplan. Instead, two floorplans are considered equivalent
if and only if their corresponding blocks have the same relative
position relationships. The formal definition is given below.

\begin{definition}\label{definition: equivalent floorplans}
{
Let $F_1$ be a floorplan with $R_1$ as its set of blocks. Let $F_2$ be another
floorplan with $R_2$ as its set of blocks. $F_1$ and $F_2$ are considered
\emph{equivalent floorplans} if and only if there is a one-to-one mapping
$g: R_1 \rightarrow R_2$ such that the following conditions hold:
\begin{enumerate}
\item For any two blocks $r,r'\in R_1$, $r$ and $r'$ share a horizontal
line segment as their common boundary with $r$ above $r'$ if and only if
$g(r)$ and $g(r')$ share a horizontal line segment as their common
boundary with $g(r)$ above $g(r')$.

\item For any two blocks $r,r'\in R_1$, $r$ and $r'$ share a vertical
line segment as their common boundary with $r$ to the left of $r'$
if and only if $g(r)$ and $g(r')$ share a vertical line segment as
their common boundary with $g(r)$ to the left of $g(r')$.
\end{enumerate}
}
\end{definition}

In Figure \ref{figure: example floorplans}, (a) and (b) have the same
number of blocks and the position relationships between their blocks
are identical. Therefore, (a) and (b) are equivalent floorplans. 
However, (c) is not equivalent to either.
\vspace{\baselineskip}

The objects of \emph{mosaic floorplans} are the same as the objects of
the floorplans. However, mosaic floorplans use a different definition
of equivalence. Informally speaking, two mosaic floorplans are considered
equivalent if and only if they can be converted to each other
by sliding the horizontal and vertical line segments.
The equivalence of the mosaic floorplans is formally
defined by using the \emph{horizontal constraint graph} and 
the \emph{vertical constraint graph} \cite{HHCG2000}.
The horizontal constraint graph describes the horizontal relationship
between the vertical line segments of a floorplan.
The vertical constraint graph describes the vertical relationship
between the horizontal line segments of a floorplan. The formal
definitions are given below.

\begin{definition} \label{definition: constraint graphs}
Let $F$ be a floorplan.
\begin{enumerate}
\item The \emph{horizontal constraint graph} $G_H(F)$ of $F$ is a directed graph.
The vertex set of $G_H(F)$ 1-to-1 corresponds to the set of the vertical line 
segments of $F$. For two vertices $u_1$ and $u_2$ in $G_H(F)$, there is a
directed edge $u_1 \rightarrow u_2$ if and only if there is a block $b$ in $F$ such that
the vertical line segment $v_1$ corresponding to $u_1$ is on the left boundary of $b$
and the vertical line segment $v_2$ corresponding to $u_2$ is on the right boundary
of $b$.
\item The \emph{vertical constraint graph} $G_V(F)$ of $F$ is a directed graph.
The vertex set of $G_V(F)$ 1-to-1 corresponds to the set of the horizontal line 
segments of $F$. For two vertices $u_1$ and $u_2$ in $G_H(F)$, there is a
directed edge $u_1 \rightarrow u_2$ if and only if there is a block $b$ in $F$ such that
the horizontal line segment $h_1$ corresponding to $u_1$ is on bottom boundary of $b$
and the horizontal line segment $h_2$ corresponding to $u_2$ is on the top boundary
of $b$.
\end{enumerate}
\end{definition}

The graphs in Figure \ref{figure: constraint graphs} are the constraint
graphs of all three floorplans shown in
Figure \ref{figure: example floorplans} (a), (b), and (c).
Note that the bottom (top, left, right, respectively) boundary of the
floorplan is represented by the south (north, west, east, respectively)
vertex labeled by S (N, W, E, respectively) in the constraint graphs.
Also note that each edge in $G_V(F)$ ($G_V(F)$, respectively) 
corresponds to a block in the floorplan. 

\begin{figure}[htbp]
\centerline
{
\includegraphics[scale=0.1]{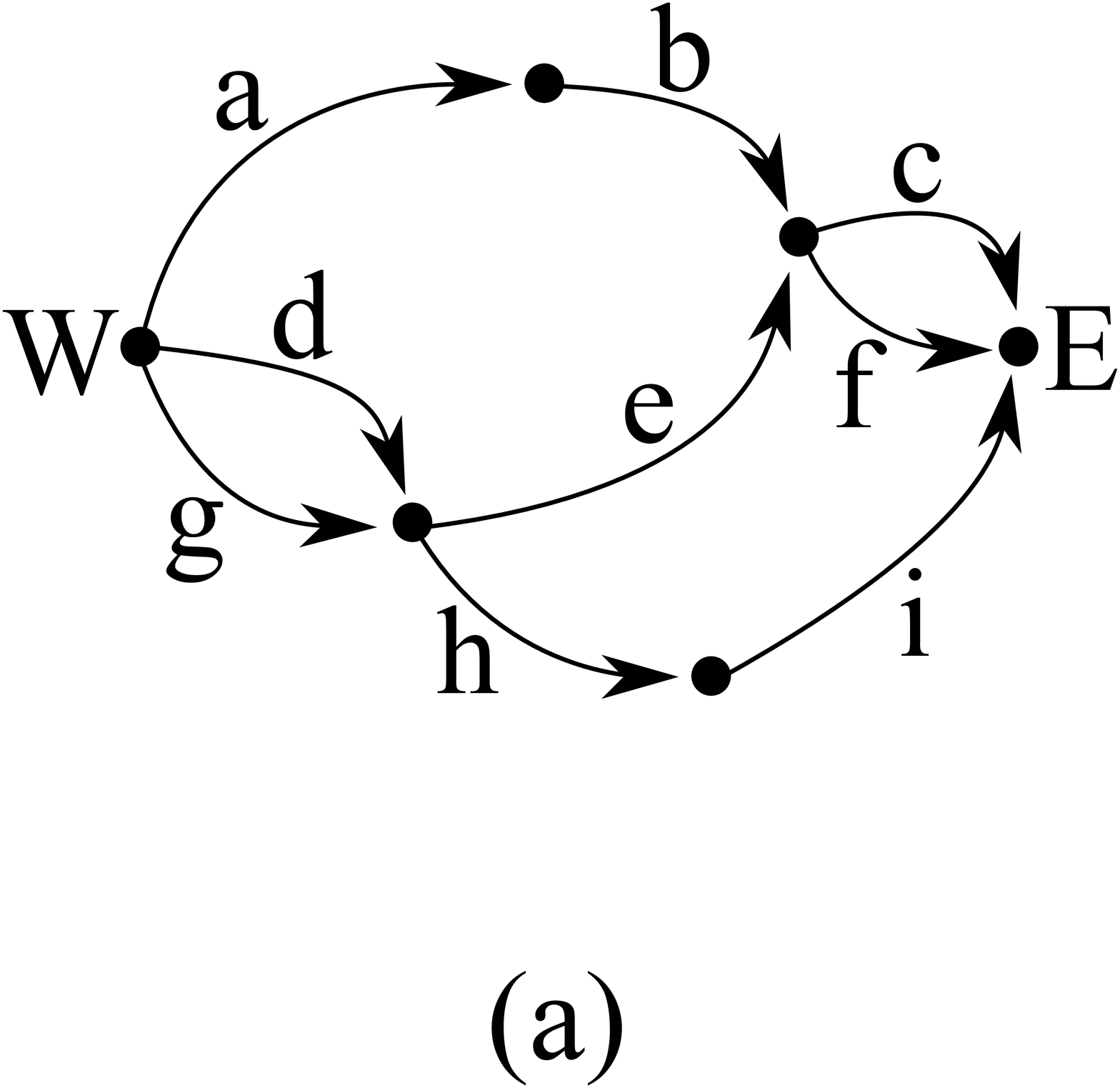}
\hspace{0.75in}
\includegraphics[scale=0.1]{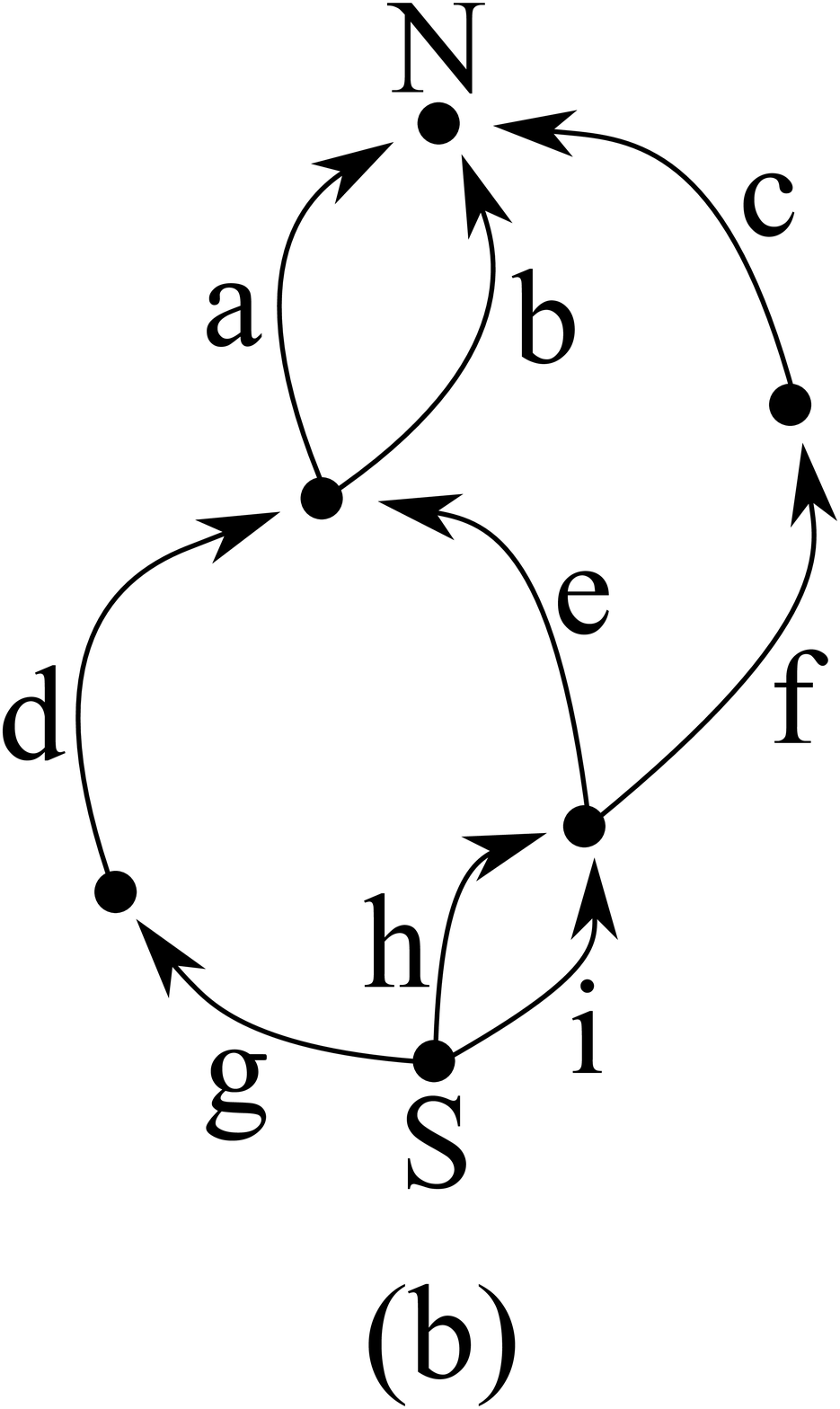}
}
\caption{The constraint graphs representing all three mosaic floorplans in Figure \ref{figure: example floorplans}. (a) is the horizontal constraint graph. (b) is the vertical constraint graph.}
\label{figure: constraint graphs}
\end{figure}

\begin{definition} \label{definition: equivalent mosaic floorplans}
Two mosaic floorplans are \emph{equivalent mosaic floorplans} if and only if they have identical horizontal constraint graphs and vertical constraint graphs.
\end{definition}

Thus all three floorplans shown in Figure \ref{figure: example floorplans} (a), (b),
and (c) are equivalent mosaic floorplans. Note that 
the floorplan in Figure \ref{figure: example floorplans} (c) is obtained from
the floorplan in Figure \ref{figure: example floorplans} (b) by sliding
the horizontal line segment between the blocks $g$ and $d$ downward;
the horizontal line segment between the blocks $f$ and $c$ upward;
the vertical line segment between the blocks $a$ and $b$ to the right.

\subsection{Applications of Floorplans and Mosaic Floorplans}\label{section: applications}

Floorplans and mosaic floorplans are used in the first major stage 
(called floorplanning) in the physical design cycle of VLSI
(Very Large Scale Integration) circuits \cite{Le1990}. The blocks in a
floorplan correspond to the components of a VLSI chip. The floorplanning 
stage is used to plan the relative position of the circuit components.
At this stage, the blocks do not have specific sizes assigned to them yet.
So only the position relationship between the blocks are considered.

For a floorplan, the wires between two blocks run cross their common boundary.
In this setting, two equivalent floorplans provide the same connectivity between blocks.
For a mosaic floorplan, the line segments are the wires. Any block with a line
segment on its boundary can be connected to the wires represented by the line segment.
In this setting, two equivalent mosaic floorplans provide the same connectivity
between blocks.


One of the main problems in this area is to find a short binary representation
of floorplans and mosaic floorplans. These representations are used by various
algorithms to generate floorplans in order to solve various VLSI layout
optimization problems. Shorter representation allows more efficient
optimization algorithms.

\subsection{Baxter Permutations}

\emph{Baxter permutations} are a set of permutations defined by prohibited
subsequences. They were first introduced in \cite{Ba1964}. 
It was shown in \cite{Gi2011} that the set of Baxter permutations
has one-to-one correspondences to many interesting objects in the
so-called \emph{Baxter combinatorial family}.
For examples, \cite{BBF2010} showed that \emph{plane bipolar
orientations} with $n$ edges have a one-to-one correspondence with
Baxter permutations of length $n$.  \cite{Ca2010} establishes a
relationship between Baxter permutations and pairs of alternating
sign matrices.


In particular, it was shown in \cite{ABP2006,DG1998,YCCG2003}
that mosaic floorplans are one of the objects in the Baxter
combinatorial family. A simple and efficient one-to-one
correspondence between mosaic floorplans and Baxter permutations
was established in \cite{ABP2006,DG1998}. 
As a result, any binary representation of mosaic floorplans can also be
converted to a binary representation of Baxter permutations.

\subsection{Previous Work on Representations of Floorplans and Mosaic Floorplans} \label{section: previous}

Because of their applications in VLSI physical design, the representations
of floorplans and mosaic floorplans have been studied extensively by
mathematicians, computer scientists and electrical engineers. Although their
definitions are similar, the combinatorial properties of floorplans
and mosaic floorplans are quite different. The following is a partial
list of previous research on floorplans and mosaic floorplans.

\vspace{\baselineskip}
\noindent \underline{Floorplans:}

There is no known formula for calculating $F(n)$, the number of $n$-block 
floorplans. The first few values of $F(n)$ (starting from $n=1$) are
$\{ 1, 2, 6, 24, 116, 642, 3938, \ldots\}$.
Researchers have been trying to bound the range of $F(n)$. 
In \cite{ANY2007}, it was shown that there exists a constant 
$c = \lim_{n \rightarrow \infty} (F(n))^{1/n}$ and $11.56 < c < 28.3$. 
This means that $11.56^n \leq F(n) \leq 28.3^n$ for large $n$. The upper
bound of $F(n)$ is reduced to $F(n) \leq 13.5^n$ in \cite{FIT2009}. 

Algorithms for generating floorplans were presented in \cite{Na2001}.
In \cite{YN2006}, a $(5n-5)$-bit binary string representation of $n$-block
floorplans was found. A different $5n$-bit binary string representation
of $n$-block floorplans was presented in \cite{YN2007}. The shortest
known binary string representation of $n$-block floorplans was given
in \cite{TFI2009}. This representation uses $(4n-4)$ bits.

Since $F(n) \geq 11.56^n$ for large $n$ \cite{ANY2007}, any binary
string representation of $n$-block floorplans must use at least 
$\log_2 11.56^n = 3.531n$ bits. Closing the gap between the known
$(4n-4)$-bit binary representation and the $3.531n$ lower bound remains an
open research problem \cite{TFI2009}.

\vspace{\baselineskip}
\noindent \underline{Mosaic Floorplans:}

It was shown in \cite{DG1998} that the set of $n$-block mosaic floorplans
has a one-to-one correspondence to the set of Baxter permutations,
and the number of $n$-block mosaic floorplans equals to the $n^{th}$
\emph{Baxter number} $B(n)$, which is defined as the following:

\begin{displaymath}
B(n) = \left( \begin{array}{c} n+1 \\ 1 \end{array} \right)^{-1}
\left( \begin{array}{c} n+1 \\ 2 \end{array} \right)^{-1}
\sum_{r=0}^{n-1} \left( \begin{array}{c} n+1 \\ r \end{array} \right) 
\left( \begin{array}{c} n+1 \\ r+1 \end{array} \right)
\left( \begin{array}{c} n+1 \\ r+2 \end{array} \right)
\end{displaymath}

In \cite{SC2003}, it was shown that $B(n) = \Theta(8^n/n^4)$. The first
few Baxter numbers (staring from $n=1$) are $\{1, 2, 6, 22, 92, 422, 2074, \ldots\}$.

There is a long list of papers on representation problem of mosaic floorplans.
\cite{MF1995} proposed a \emph{sequence pair} (SP) representation.
Two sets of permutations are used to represent the position relations
between blocks. The length of the representation is $2n\log_2 n$ bits.

\cite{HHCG2000} proposed a \emph{corner block list} (CB) representation
for mosaic floorplans. The representation consists of a list $S$ of blocks,
a binary string $L$ of $(n-1)$ bits, and a binary string $T$ of $2n-3$
bits. The total length of the representation is $(3n + n\log_2 n)$ bits.

\cite{YCS2003} proposed a \emph{twin binary sequences} (TBS)
representation for mosaic floorplans. The representation consists of
4 binary strings $(\pi, \alpha, \beta, \beta')$, where $\pi$ is a
permutation of integers $\{ 1, 2,\ldots, n\}$, and the other three
strings are $n$ or $(n-1)$ bits long. The total length of the
representation is $3n+n\log_2 n$.

A common feature of above representations is that each block in the
mosaic floorplan is given an explicit name (such as an integer
between 1 and $n$). They all use at least one list (or permutation)
of these names in the representation. Because at least $\log_2 n$
bits are needed to represent every integer in the range $[1,n]$,
the length of these representations is inevitably at least
$n\log_2 n$ bits.

A different approach was introduced in \cite{YCCG2003}. They use a
pair of \emph{twin pair binary trees} $t_1$ and $t_2$ to represent
mosaic floorplans. The blocks of the mosaic floorplan are not given
explicit names. Rather, the shape of the two trees $t_1$ and $t_2$
are used to encode the position relations of blocks. In this
representation, each tree consists of $2n$ nodes. Thus, each tree
can be encoded by using $4n$ bits. So the total length of the
representation is $8n$ bits. They also proposed an alternate
representation using a pair of $n$-node trees. However, the nodes
in the two trees are given names, and the length of the
representation is at least $2n\log_2 n$.

In \cite{SKM2003}, a representation called \emph{quarter-state-sequence}
(QSS) was presented. It uses a $Q$ sequence that represents the
configuration of one of the corners of the mosaic floorplan. The length
of the $Q$ sequence representation is $4n$ bits.
This is the best known representation for mosaic floorplans.

Because the number of $n$-block mosaic floorplans equals the $n^{th}$
Baxter number, at least $\log_2 B(n) = \log_2 \Theta(8^n/n^4) = 3n-o(n)$
bits are needed to represent mosaic floorplans.

\subsection{Our Main Result}

\begin{theorem}\label{theorem: main}
The set of $n$-block mosaic floorplans can be represented by
$(3n-3)$ bits, which is optimal up to an additive lower order term.
\end{theorem}

Most binary representations of mosaic floorplans discussed in section
\ref{section: previous} are fairly complex. In contrast,
the representation introduced in this paper is \emph{very simple}
and easy to implement.

By using the simple one-to-one correspondence between mosaic floorplans
and Baxter permutations described in \cite{ABP2006},
the methods presented in this paper also work on Baxter permutations.
Hence, our optimal representation of mosaic floorplans also leads to an
optimal representation of Baxter permutations and 
all other objects in the Baxter combinatorial family.

\section{Optimal Binary Representation of Mosaic Floorplans}

In this section, we describe our optimal representation of mosaic floorplans. 

\subsection{Standard Form of Mosaic Floorplans}

In the following, we introduce the notion of \emph{standard form} of
mosaic floorplans, which plays a central role in our
representation.

Let $M$ be a mosaic floorplan. Let $h$ be a horizontal line segment in
$M$. The \emph{upper segment set} of $h$ and the \emph{lower segment set}
of $h$ are defined as the following:

\vspace{\baselineskip}
\noindent ABOVE$(h)=$ the set of vertical line segments of $M$ that intersect $h$ and are above $h$.

\noindent BELOW$(h)=$ the set of vertical line segments of $M$ that intersect $h$ and are below $h$.

\vspace{\baselineskip}
Similarly, for a vertical line segment $v$ in $M$, the \emph{left segment set} of $v$ and the \emph{right segment set} of $h$ are defined as the following:

\vspace{\baselineskip}
\noindent LEFT$(v)=$ the set of horizontal segments of $M$ that intersect $v$ and are on the left of $v$.

\noindent RIGHT$(v)=$ the set of horizontal segments of $M$ that intersect $v$ and are on the right of $v$.
\vspace{\baselineskip}

\begin{figure}[htbp]
\centerline
{
\includegraphics[scale=0.14]{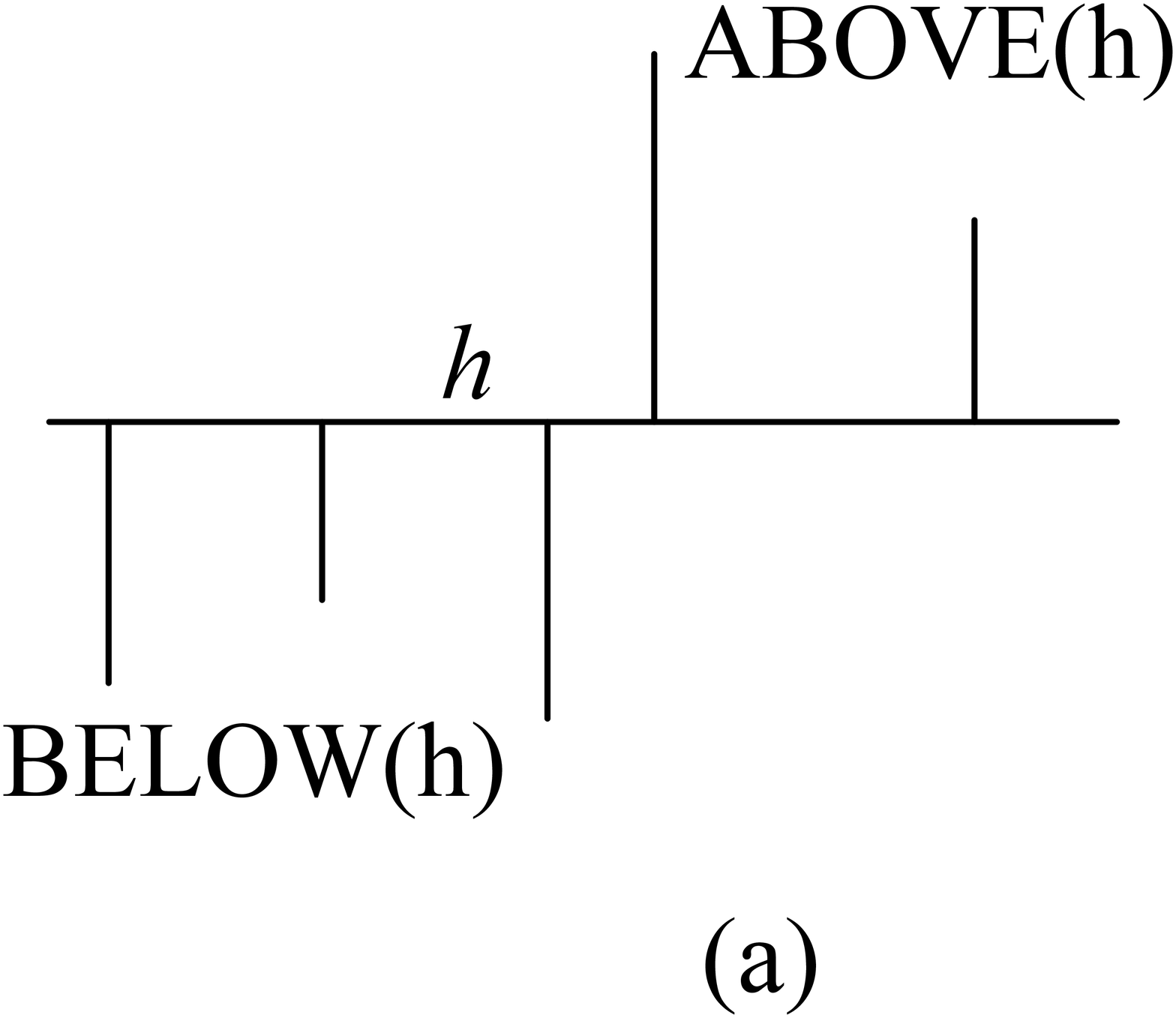}
\hspace{0.15in}
\includegraphics[scale=0.14]{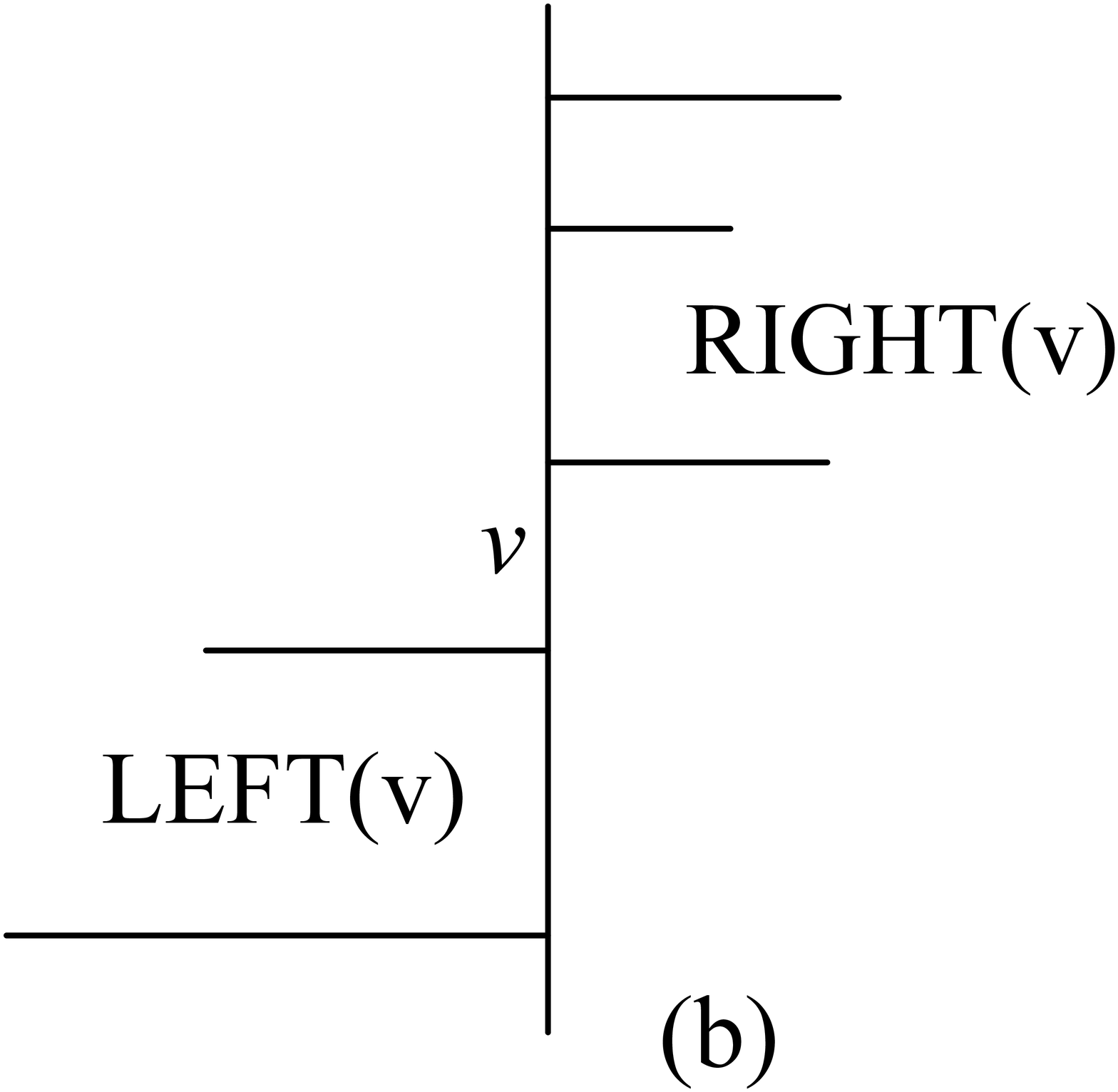}
}
\caption{Standard form of mosaic floorplans.}
\label{figure: standard form}
\end{figure}

\begin{definition}
{
A mosaic floorplan $M$ is in \emph{standard form} if the following hold:
\begin{enumerate}
\item For every horizontal segment $h$ in $M$, all vertical segments in ABOVE$(h)$ appear to the right of all vertical segments in BELOW$(h)$. (See Figure \ref{figure: standard form}(a).)
\item For every vertical segment $v$ in $M$, all horizontal segments in RIGHT$(v)$ appear above all horizontal segments in LEFT$(v)$. (See Figure \ref{figure: standard form}(b).)
\end{enumerate}
}
\end{definition}

The mosaic floorplan shown in Figure \ref{figure: example floorplans} (c) is the standard form of mosaic floorplans shown in Figure \ref{figure: example floorplans} (a) and Figure \ref{figure: example floorplans} (b).

The standard form $M_{\rm standard}$ of a mosaic floorplan $M$ can be obtained
by sliding its vertical and horizontal line segments. Because of the 
equivalence definition of mosaic floorplans, $M_{\rm standard}$ and $M$
are considered the same mosaic floorplans. For a given $M$, $M_{\rm standard}$
can be obtained in linear time by using the horizontal constraint graphs and
vertical constraint graphs described in \cite{HHCG2000}. From now on, all
mosaic floorplans are assumed to be in standard form.

\subsection{Staircases}

\begin{definition} \label{definition: staircase}
{ 
A \emph{staircase} is an object that satisfies the following conditions:

\begin{enumerate}
\item The border contains a line segment on the $x$-axis and a line segment on
the $y$-axis.
\item The remainder of the border is a non-increasing line segments consisting
of vertical and horizontal line segments.
\item The interior is divided into smaller rectangular subsections by
horizontal and vertical line segments.
\item No four subsections meet at the same point.
\end{enumerate}
}
\end{definition}

\begin{figure}[htbp]
\centerline
{
\includegraphics[scale=0.1]{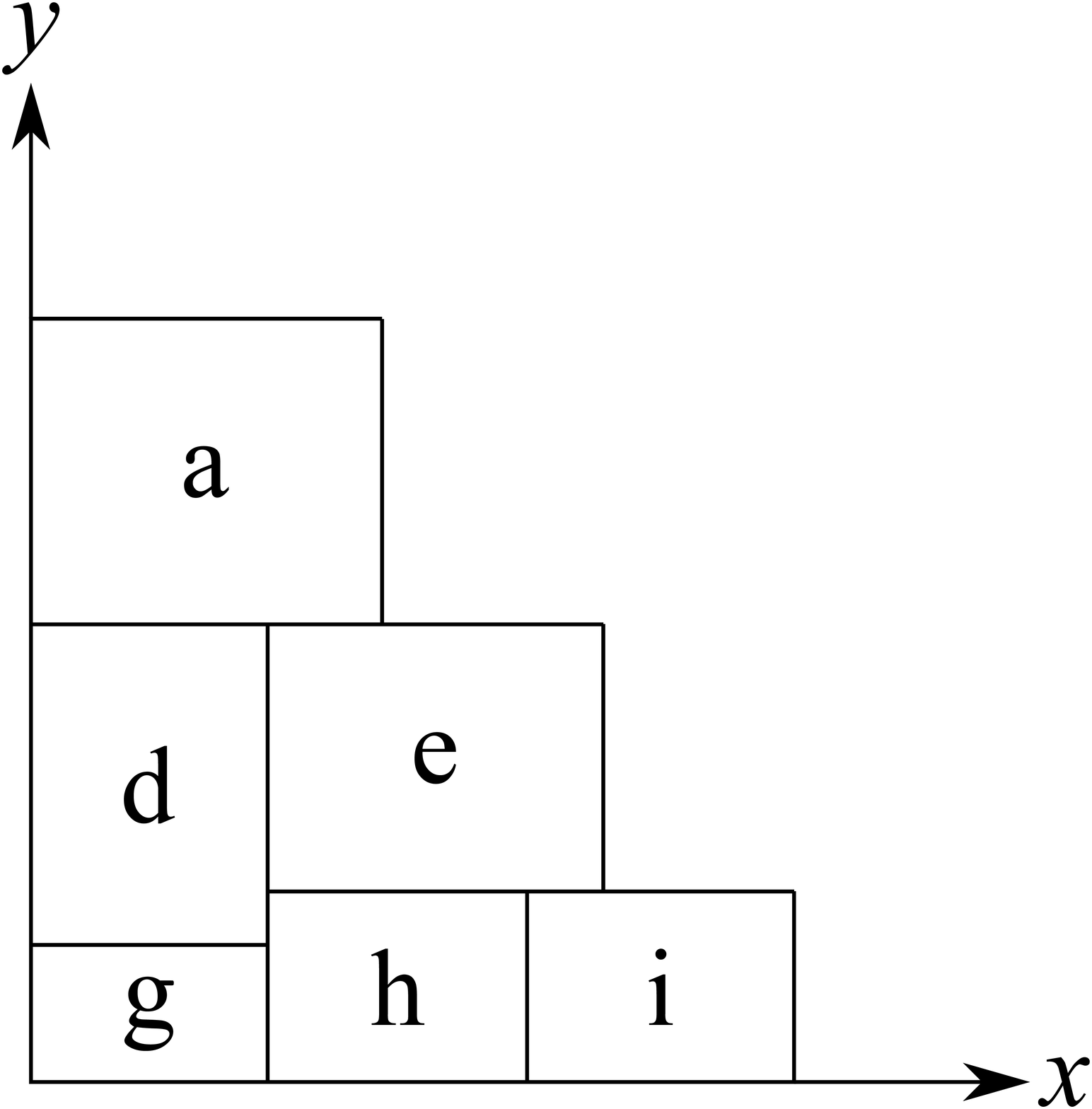}
}
\caption{A staircase with $n=6$ blocks and $m=3$ steps that is obtained
from the mosaic floorplan in Figure \ref{figure: example floorplans} (c) by
deleting the blocks $b,c$ and $f$.}
\label{figure: staircase}
\end{figure}

A \emph{step} of a staircase $S$ is a horizontal line segment on the
border of $S$, excluding the $x$-axis. Figure \ref{figure: staircase} shows
a staircase with $n=6$ blocks and $m=3$ steps. Note that a mosaic
floorplan is just a special case of a staircase with $m=1$ step.

\subsection{Deletable Rectangles}

\begin{definition}
{ 
A \emph{deletable rectangle} of a staircase $S$ is a block that satisfies
the following conditions:
\begin{enumerate}
\item Its top edge is completely contained in the border of $S$.
\item Its right edge is completely contained in the border of $S$.
\end{enumerate}
}
\end{definition}

In the staircase shown in
Figure \ref{figure: staircase}, the block $a$ is the only deletable rectangle.
The concept of deletable rectangles is a key idea for the methods introduced
in this paper. This concept was originally defined in \cite{TFI2009} for
their $(4n-4)$-bit representation of floorplans. However, a modified
definition of deletable rectangles is used in this paper to create a
$(3n-3)$-bit representation of mosaic floorplans. 

\begin{lemma} \label{lemma: removal}
The removal of a deletable rectangle from a staircase results in another
staircase unless the original staircase contains only one block.
\end{lemma}
\begin{proof}
Let $S$ be a staircase with more than one block and let $r$ be a
deletable rectangle in $S$. Define $S'$ to be the object that
results when $r$ is removed from $S$. Because the removal of $r$
still leaves $S'$ with at least one block, the border of $S'$
still contains a line segment on the $x$-axis and a line segment on
the $y$-axis, so condition (1) of a staircase holds for $S'$. 
Removing $r$ will not cause the remainder of the border to have
an increasing line segment because the right edge of $r$ must be
completely contained in the border, so condition (2) of a staircase
also holds for $S'$. The removal of $r$ does not form new line
segments, so the interior of $S'$ will still be
divided into smaller rectangular subsections by vertical and
horizontal line segments, and no four subsections in $S'$ will
meet at the same point. Thus, conditions (3) and (4) of a staircase
hold for $S'$. Therefore, $S'$ is a staircase.
\end{proof}

We can now outline the basic ideas of our representation.
Given a mosaic floorplan $M$, we remove deletable rectangles of $M$ one
by one. By Lemma \ref{lemma: removal}, this results in a sequence of
staircases, until only one block remains. We record necessary location
information of these deletable rectangles (which will be the binary
representation of $M$) so that we can reconstruct the original floorplan $M$.
However, if there are multiple deletable rectangles for these staircases,
we will have to use more bits than we can afford. Fortunately, the
following key lemma shows that this does not happen. 

\begin{lemma} \label{lemma: unique}
Let $M$ be a $n$-block mosaic floorplan in standard form. Let $S_n=M$,
and let $S_{i-1}$ ($2 \leq i \leq n$) be the staircase
obtained by removing a deletable rectangle $r_i$ from $S_i$.

\begin{enumerate}
\item There is a single, unique deletable rectangle in $S_i$ for $1 \leq i \leq n$.
\item $r_{i-1}$ is adjacent to $r_i$ for $2 \leq i \leq n$.
\end{enumerate}
\end{lemma}

\begin{proof}
The proof is by reverse induction.

Clearly, $S_n=M$ has only one deletable rectangle located at the
top right corner of $M$.

Assume that $S_{i+1}$ ($i \leq n-1$) has exactly one deletable rectangle
$r_{i+1}$. Let $h$ be the horizontal line segment in $S_{i+1}$ that contains
the bottom edge of $r_{i+1}$, and let $v$ be the vertical line
segment in $S_{i+1}$ that contains the left edge of $r_{i+1}$
(see Figure \ref{figure: unique}).
Let $a$ be the uppermost block in $S_{i+1}$ whose right edge aligns with
$v$, and let $b$ be the rightmost block in $S_{i+1}$
whose top edge aligns with $h$. After $r_{i+1}$ is removed from $S_{i+1}$,
$a$ and $b$ are the only candidates for deletable rectangles of 
the resulting staircase $S_i$. There are two cases:

\begin{figure}[htbp]
\centerline
{
\includegraphics[scale=0.12]{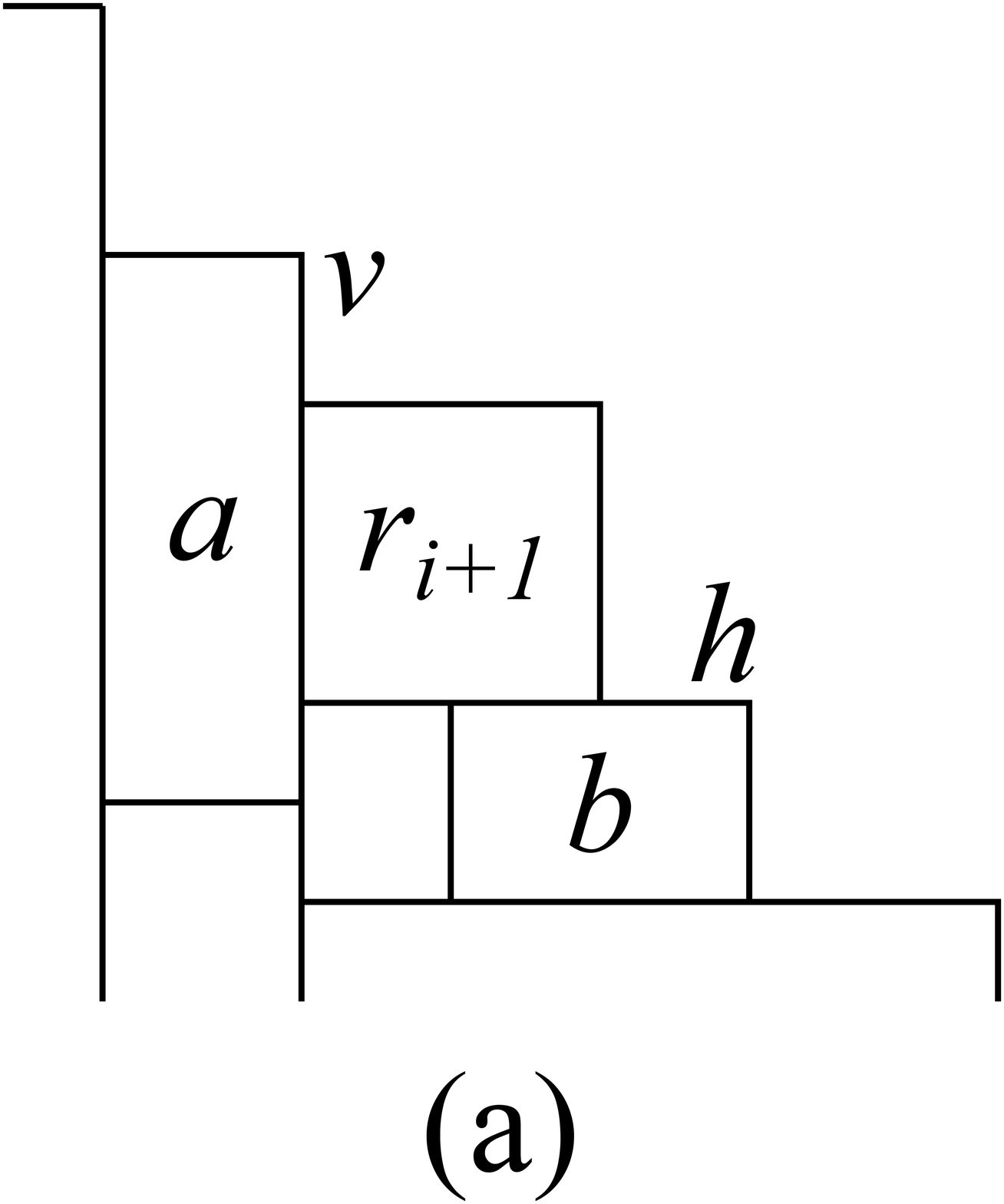}
\hspace{0.5in}
\includegraphics[scale=0.12]{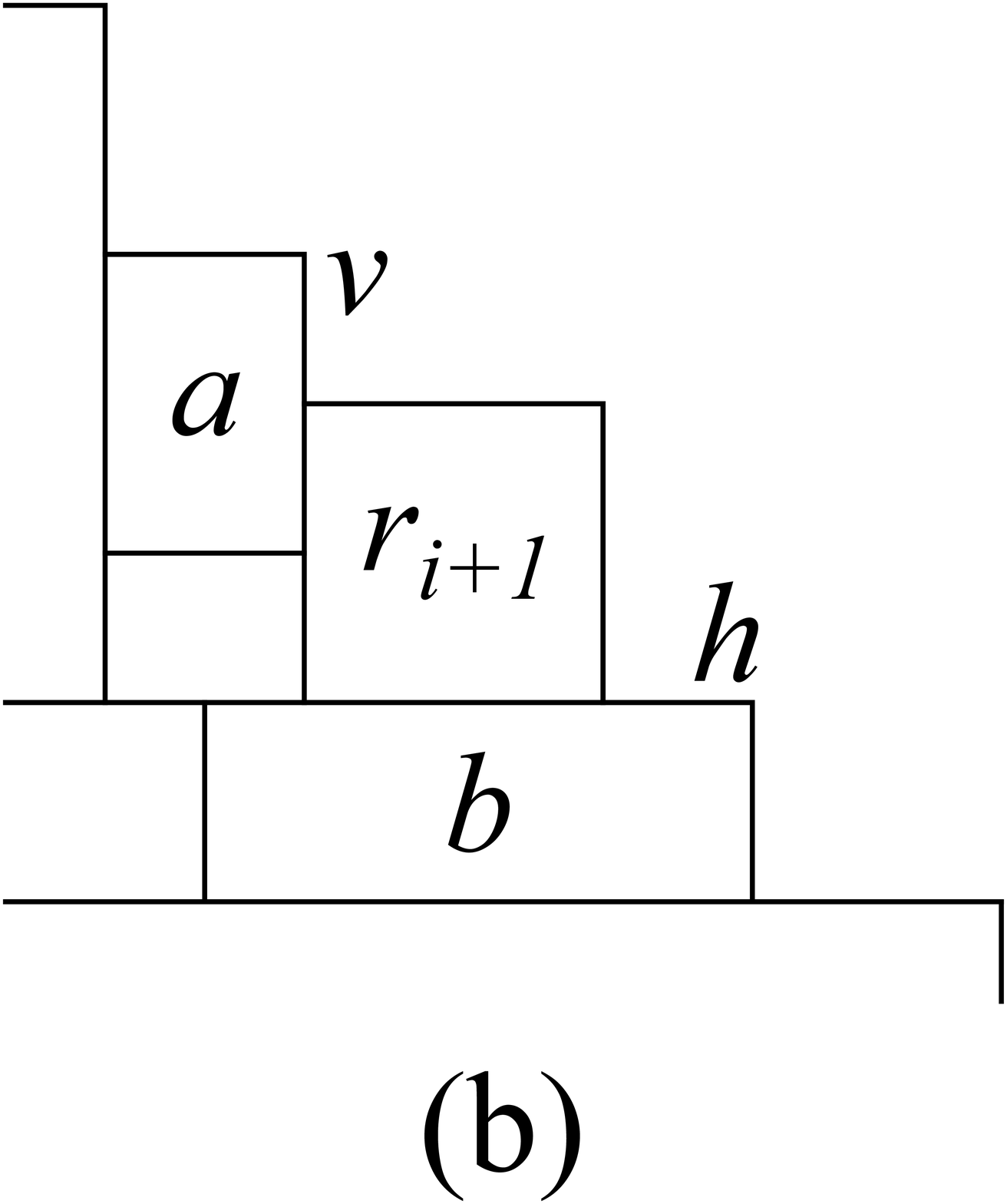}
}
\caption{Proof of Lemma \ref{lemma: unique}.}
\label{figure: unique}
\end{figure}

\begin{enumerate}
\item The line segments $h$ and $v$ form a $\vdash$-junction
(see Figure \ref{figure: unique} (a).) Then, the bottom edge of $a$ must
be below $h$ because $M$ is a standard mosaic floorplan, and $a$ is not
a deletable rectangle in $S_i$. Thus, the block $b$ is the only deletable
rectangle in $S_i$.

\item The line segments $h$ and $v$ form a $\perp$-junction (see
Figure \ref{figure: unique} (b).) Then, the left edge of $b$ must be
to the left of $v$ because $M$ is a standard mosaic floorplan, and
$b$ is not a deletable rectangle in $S_i$. Thus, the block $a$
is the only deletable rectangle in $S_i$.
\end{enumerate}

In both cases, only one deletable rectangle $r_i$ (which is either
$a$ or $b$) is revealed when the deletable rectangle $r_{i+1}$ is
removed. Because there is only one deletable rectangle in $S_n=M$,
all subsequent staircases contain exactly one deletable rectangle.
Thus, (1) is true. Also, $r_{i+1}$ is adjacent to $r_i$ in both
cases, so (2) is true.
\end{proof}

\vspace{\baselineskip}
Let $S$ be a staircase and $r$ be a deletable rectangle of $S$
whose top side is on the $k$-th step of $S$. There are four types
of deletable rectangles.

\begin{figure}[htbp]
\centerline
{
\includegraphics[scale=0.11]{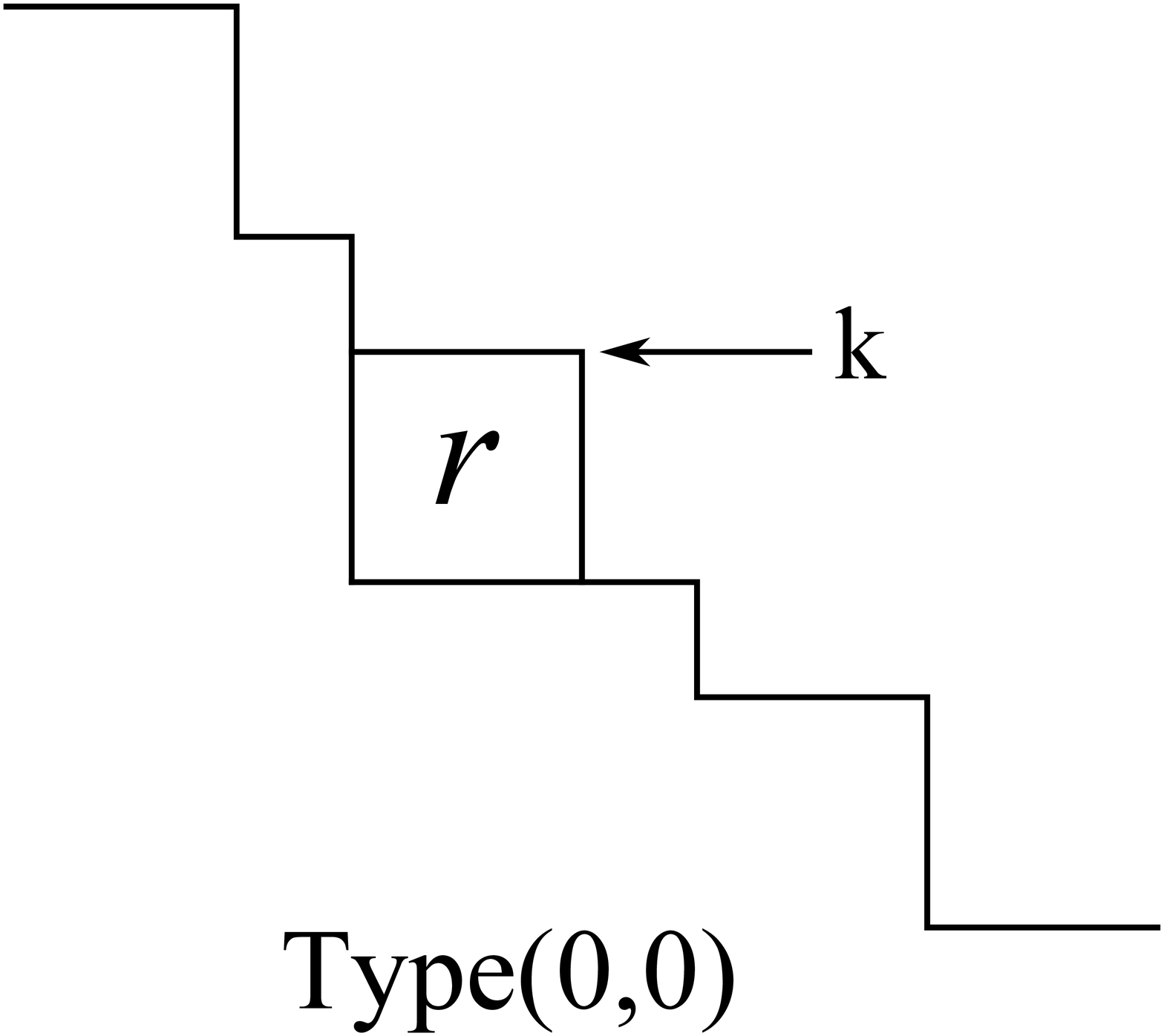}
\hspace{0.4in}
\includegraphics[scale=0.11]{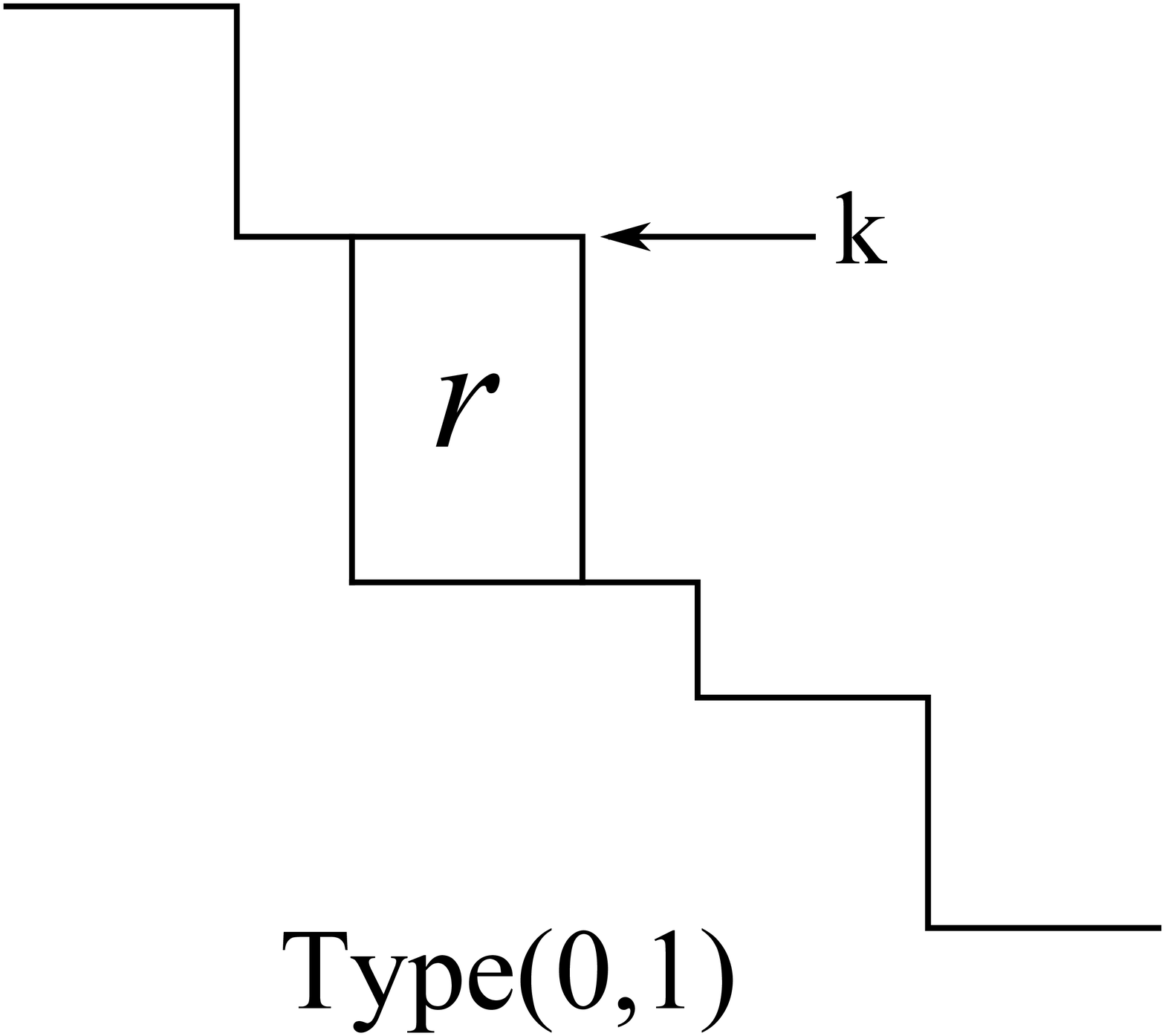}
\hspace{0.4in}
\includegraphics[scale=0.11]{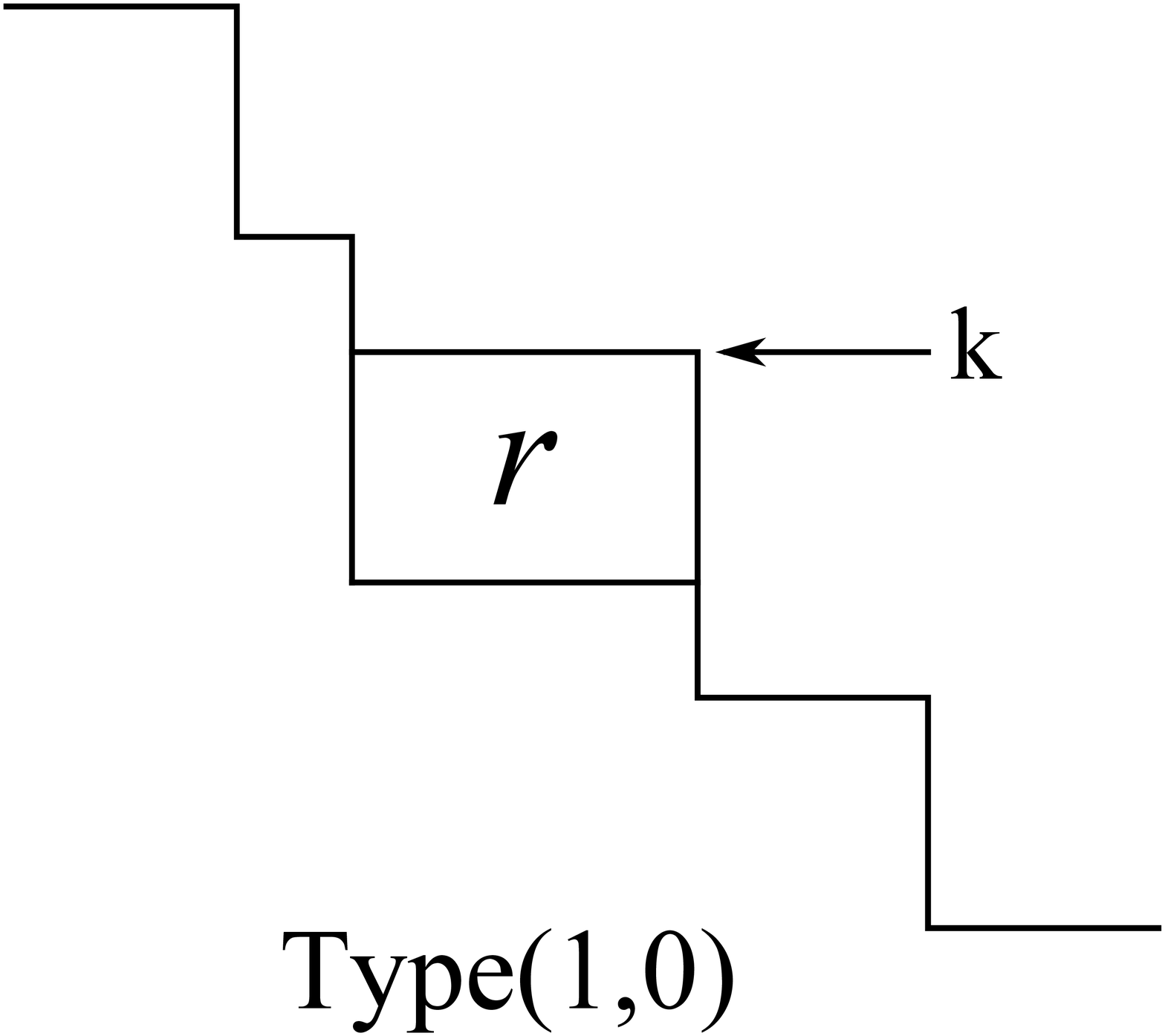}
\hspace{0.4in}
\includegraphics[scale=0.11]{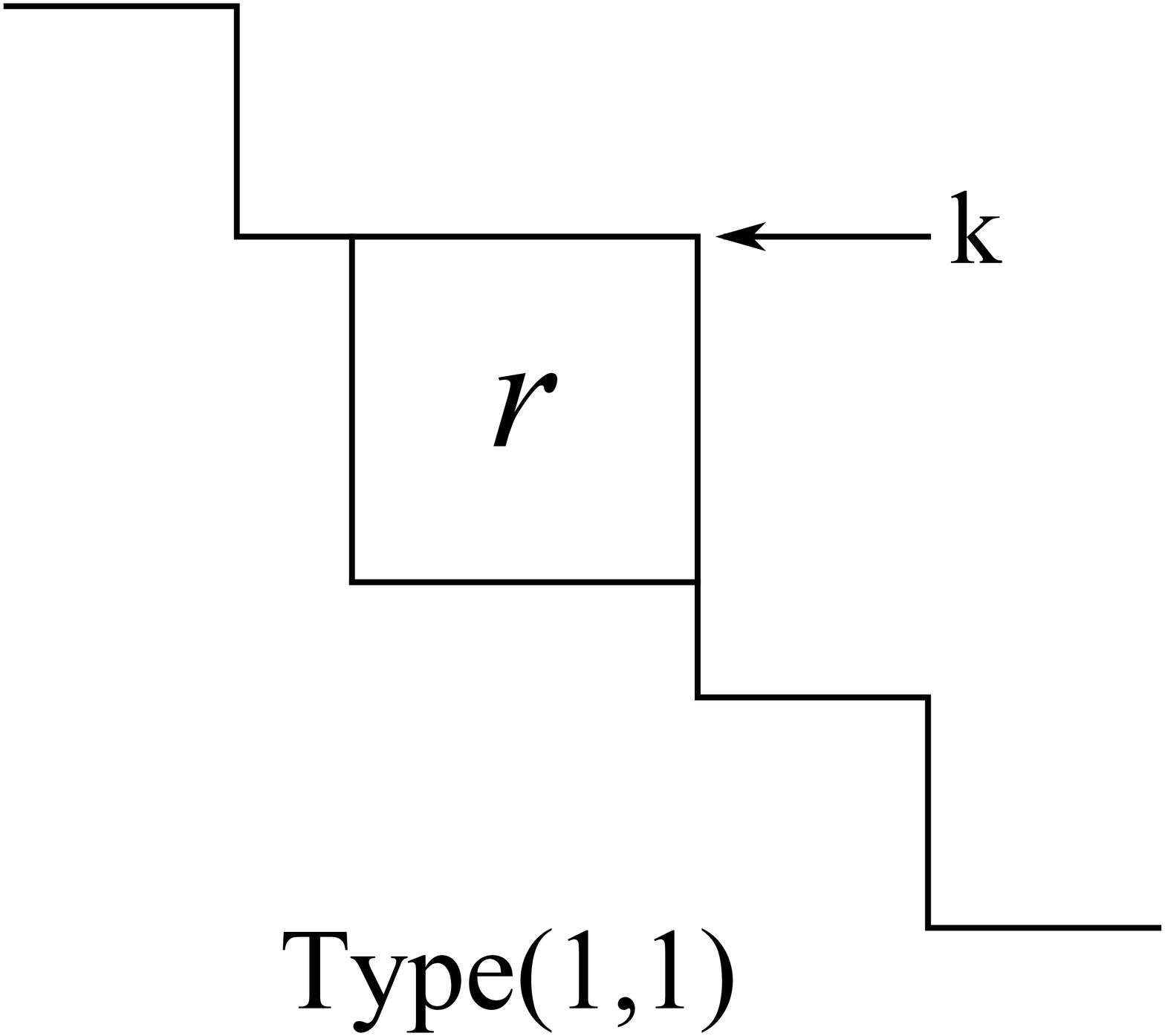}
}
\caption{The four types of deletable rectangles.}
\label{figure: types}
\end{figure}

\begin{enumerate}

\item Type $(0,0)$:
\begin{enumerate}
\item The top side of $r$ is the entire $k$-th step.
\item The right side of $r$ intersects the $(k-1)$-th step.
\item The deletion of $r$ decreases the number of steps by one.
\end{enumerate}

\item Type $(0,1)$:
\begin{enumerate}
\item The top side of $r$ is only a part of the $k$-th step.
\item The right side of $r$ intersects the $(k-1)$-th step.
\item The deletion of $r$ does not change the number of steps.
\end{enumerate}

\item Type $(1,0)$:
\begin{enumerate}
\item The top side of $r$ is the entire $k$-th step.
\item The right side of $r$ is only a part of the right side of the $k$-th
step (namely the right-bottom corner of $r$ is a $\dashv$ shape junction).
\item The deletion of $r$ does not change the number of steps.
\end{enumerate}

\item Type $(1,1)$:
\begin{enumerate}
\item The top side of $r$ is only a part of the $k$-th step.
\item The right side of $r$ is only a part of the right side of the $k$-th
step (namely the right-bottom corner of $r$ is a $\dashv$ shape junction).
\item The deletion of $r$ increases the number of steps by one.
\end{enumerate}

\end{enumerate}

\subsection{Optimal Binary Representation}

Our binary representation of mosaic floorplans
depends on the fact that a mosaic floorplan $M$ is a special case of
a staircase and the fact that the removal of a deletable rectangle from
a staircase results in another staircase. The binary string used to
represent $M$ records the unique sequence of deletable rectangles that
are removed in this process. The information stored by this binary string
enable us to reconstruct the original mosaic floorplan $M$.

A 3-bit binary string is used to record the information for each deletable
rectangle $r_i$. The string has two parts: The type and the location of
$r_i$. To record the type of $r_i$, the bits corresponding to its type is
stored directly. To store the location, we note that, by Lemma \ref{lemma: unique},
two consecutive deletable rectangles $r_i$ and $r_{i-1}$ are adjacent to each other.
Thus, they must share either a horizontal edge or a vertical edge. 
A single bit can be used to record the location of $r_i$ with respect to
$r_{i-1}$: a 1 if they share a horizontal edge, and a 0 if they
share a vertical edge.

\vspace{\baselineskip}
\noindent \textbf{Encoding Procedure:}

Let $M$ be the $n$-block mosaic floorplan to be encoded. Starting from
$S_n=M$, remove the unique deletable rectangles $r_i$, where
$2 \leq i \leq n$, one by one. For each deletable rectangle
$r_i$, two bits are used to record the type of $r_i$, and one bit is
used to record the type of the common boundary shared by $r_i$ and $r_{i-1}$.

\begin{figure}[h]
\centerline
{
\includegraphics[scale=0.6]{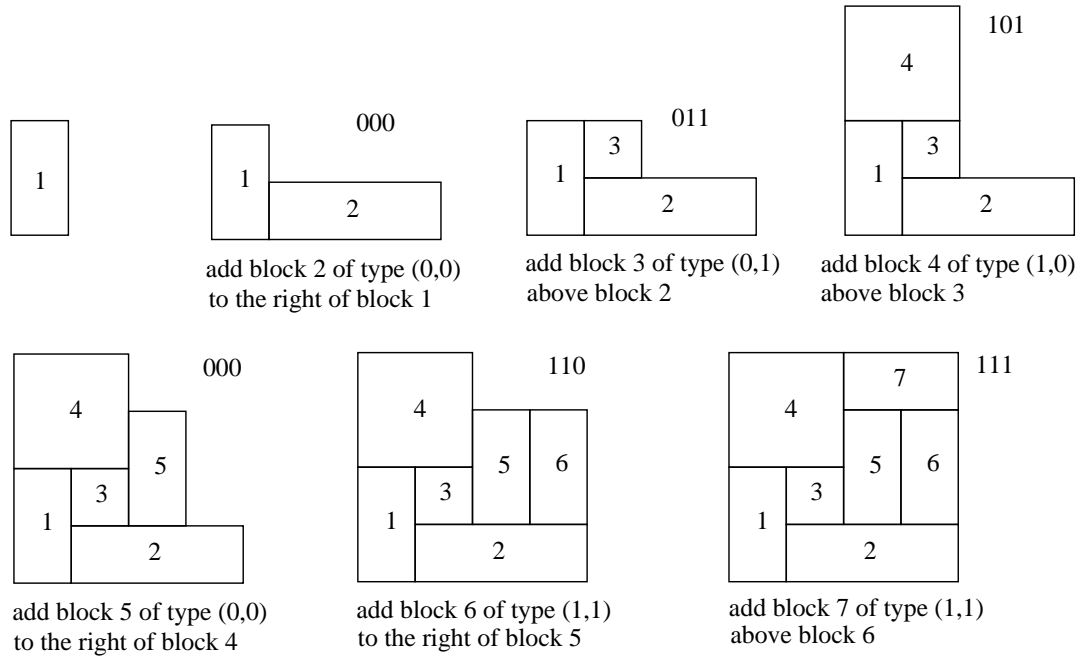}
}
\caption{An example of the presentation.}
\label{figure: example}
\end{figure}

\vspace{\baselineskip}
\noindent \textbf{Decoding Procedure:}

The decoding procedure simply reverses the process of removing deletable
rectangles. The process starts with the staircase $S_1$, which is a
single rectangle. Each staircase $S_{i+1}$ can be reconstructed from
the staircase $S_i$ by using the three-bit code for the deletable
rectangle $r_{i+1}$. The three-bit code records the type of $r_{i+1}$
and the type of edge shared by $r_i$ and $r_{i+1}$, so $r_{i+1}$ can be
uniquely added to $S_i$. Thus, the decoding procedure can reconstruct
original mosaic floorplan $S_n=M$.

Figure \ref{figure: example} show an example of the reconstruction
of a mosaic floorplan from its representation:
\begin{center}
000 011 101 000 110 111
\end{center}

The lower left block of the mosaic floorplan $M$ (which is the only
block of $S_1$) does not need any information to be recorded.
Each of the other blocks of $M$ needs three bits. Thus the total length
of the binary representation of $M$ is $(3n-3)$ bits.
This completes the proof of Theorem \ref{theorem: main}.

\vspace{\baselineskip}

\section{Conclusion}

In this paper, we introduced a binary representation of $n$-block
Mosaic floorplans. The representation uses $(3n-3)$ bits.
Since any representation of $n$-block mosaic floorplans requires
at least $(3n-o(n))$ bits \cite{SC2003}, our
representation is optimal (up to an additive lower term).
Our representation is very simple and easy to implement.

Mosaic floorplans are known to have a simple one-to-one
correspondence with Baxter permutations. So the method used to
represent mosaic floorplans in this paper also lead to an optimal $(3n-3)$ bits
representation of Baxter permutation of length $n$, and all objects in
the Baxter combinatorial family.

\bibliographystyle{plain} 


\begin{thebibliography}{10}

\bibitem{ABP2006}
Eyal Ackerman, Gill Barequet, and Ron~Y. Pinter.
\newblock A bijection between permutations and floorplans, and its
  applications.
\newblock {\em Discrete Applied Mathematics}, 154:1674--1684, 2006.

\bibitem{ANY2007}
Kazuyuki Amano, Shin'ichi Nakano, and Katsuhisa Yamanaka.
\newblock On the number of rectangular drawings: Exact counting and lower and
  upper bounds.
\newblock {\em IPSJ SIG Notes 2007-AL-115-5C}, pages 33--40, 2007.

\bibitem{Ba1964}
G.~Baxter.
\newblock On fixed points of the composite of commuting functions.
\newblock In {\em Proceedings American Mathematics Society 15}, pages 851--855,
  1964.

\bibitem{BBF2010}
Nicloas Bonichon, Mirelle Bousquet-M\'{e}lou, and \'{E}ric Fusy.
\newblock Baxter permutations and plane bipolar orientations.
\newblock {\em S\'{e}minaire Lotharingien de Combinatoire}, 61A, 2010.

\bibitem{Ca2010}
Hal Canary.
\newblock Aztec diamonds and baxter permutations.
\newblock {\em The Electronic Journal of Combinatorics}, 17, 2010.

\bibitem{DG1998}
S.~Dulucq and O.~Guibert.
\newblock Baxter permutations.
\newblock {\em Discrete Mathematics}, 180:143--156, 1998.

\bibitem{FIT2009}
Ryo Fujimaki, Youhei Inoue, and Toshihiko Takahashi.
\newblock An asymptotic estimate of the numbers of rectangular drawings or
  floorplans.
\newblock In {\em Proceedings 2009 IEEE International Symposium on Circuits and
  Systems}, pages 856--859, 2009.

\bibitem{Gi2011}
Samuele Giraudo.
\newblock Algebraic and combinatorial structures on baxter permutations.
\newblock {\em Discrete Mathematics and Theoretical Computer Science (DMTCS)},
  2011.

\bibitem{HHCG2000}
Xianlong Hong, Gang Huang, Yici Cai, Jiangchun Gu, Sheqin Dong, Chung-Kuan
  Cheng, and Jun Gu.
\newblock Corner-block list: An effective and efficient topological
  representation of non-slicing floorplan.
\newblock In {\em Proceedings of the International Conference on Computer Aided
  Design, (ICCAD'00)}, pages 8--12, 2000.

\bibitem{Le1990}
Thomas Lengauer.
\newblock {\em Combinatorial Algorithms for Integrated Circuit Layout}.
\newblock John Wiley \& Sons, 1990.

\bibitem{MF1995}
Hiroshi Murata and Kunihiro Fujiyoshi.
\newblock Rectangle-packing-based module placement.
\newblock In {\em Proceedings of the International Conference on Computer Aided
  Design, (ICCAD'95)}, pages 472--479, 1995.

\bibitem{Na2001}
Shin'ichi Nakano.
\newblock Enumerating floorplans with $n$ rooms.
\newblock In {\em Proceedings 12th International Symposium on Algorithms and
  Computation, (ISAAC'01). Lecture Notes in Computer Science Vol 2223}, pages
  107--115, 2001.

\bibitem{SKM2003}
Keishi Sakanushi, Yoji Kajitani, and Dinesh~P. Mehta.
\newblock The quarter-state-sequence floorplan representation.
\newblock {\em IEEE Transactions on Circuits and Systems - I: Fundamental
  Theory and Applications}, 50(3):376--386, 2003.

\bibitem{SC2003}
Zion~Cien Shen and Chris C.~N. Chu.
\newblock Bounds on the number of slicing, mosaic, and general floorplans.
\newblock {\em IEEE Transactions on Computer-Aided Design of Integrated
  Circuits and Systems}, 22(10):1354--1361, 2003.

\bibitem{TFI2009}
Toshihiko Takahashi, Ryo Fujimaki, and Youhei Inoue.
\newblock A $(4n-4)$-bit representation of a rectangular drawing or floorplan.
\newblock In {\em Proceedings 15th International Computing and Combinatorics
  Conference (COCOON'09). Lecture Notes in Computer Science Vol 5609}, pages
  47--55, 2009.

\bibitem{YN2006}
Katsuhisa Yamanaka and Shin'ichi Nakano.
\newblock Coding floorplans with fewer bits.
\newblock {\em IEICE Transactions Fundamentals}, E89(5):1181--1185, 2006.

\bibitem{YN2007}
Katsuhisa Yamanaka and Shin'ichi Nakano.
\newblock A compact encoding of rectangular drawings with efficient query
  supports.
\newblock In {\em Proceedings, 3rd International Conference on Algorithmic
  Aspects in Information and Management(AAIM'07). Lecture Notes in Computer
  Science Vol 4508}, pages 68--81, 2007.

\bibitem{YCCG2003}
Bo~Yao, Hongyu Chen, Chung-Kuan Cheng, and Ronald Graham.
\newblock Floorplan representation: Complexity and connections.
\newblock {\em ACM Transactions on Design Automation of Electronic Systems},
  8(1):55--80, 2003.

\bibitem{YCS2003}
Evangeline E.~Y. Young, N.~Chu Chris~C, and Zion~Cien Shen.
\newblock Twin binary sequences: A nonredundant representation for general
  nonslicing floorplan.
\newblock {\em IEEE Transactions on Computer-Aided Design of Integrated
  Circuits and Systems}, 22(4):457--469, 2003.

\end{thebibliography}

\end{document}